\documentclass{article}
\usepackage{graphicx} 
\usepackage[left=1in,right=1in,top=1in,bottom=1in]{geometry}
\usepackage{amssymb,amsfonts,amsmath,amsthm}
\usepackage[numbers,compress]{natbib}

\usepackage{xcolor}
\usepackage{url}
\usepackage[hidelinks]{hyperref}

\usepackage{braket}

\newtheorem{theorem}{Theorem}
\newtheorem{definition}[theorem]{Definition}
\newtheorem{remark}[theorem]{Remark}

\newcommand{\C}{\mathbb{C}}

\newcommand{\cA}{\mathcal{A}}
\newcommand{\cC}{\mathcal{C}}
\newcommand{\cD}{\mathcal{D}}
\newcommand{\cG}{\mathcal{G}}
\newcommand{\cH}{\mathcal{H}}
\newcommand{\cL}{\mathcal{L}}
\newcommand{\cR}{\mathcal{R}}
\newcommand{\cS}{\mathcal{S}}
\newcommand{\cU}{\mathcal{U}}

\renewcommand{\>}{\rangle}
\newcommand{\<}{\langle}

\title{Controlization Schemes Based on Orthogonal Arrays}
\author{Anirban Chowdhury\quad Ewout van den Berg\quad Pawel
  Wocjan\\[1ex]IBM Quantum, IBM Thomas J. Watson Research Center,
  Yorktown Heights, NY, USA}

\begin{document}

\maketitle

\abstract{ Realizing controlled operations is fundamental to the
  design and execution of quantum algorithms. In quantum simulation
  and learning of quantum many-body systems, an important subroutine
  consists of implementing a controlled Hamiltonian
  time-evolution. Given only black-box access to the uncontrolled
  evolution $e^{-iHt}$, \emph{controlizing} it, i.e., implementing
  $\mathrm{ctrl}(e^{-iHt}) = |0\>\<0|\otimes I + |1\>\<1 |\otimes
  e^{-iHt}$ is non-trivial.  Controlization has been recently used in
  quantum algorithms for transforming unknown Hamiltonian
  dynamics~\cite{odake2024universal} leveraging a scheme introduced in
  Refs.~\cite{nakayama2015energy,dong2021controlled}.  The main idea
  behind the scheme is to intersperse the uncontrolled evolution with
  suitable operations such that the overall dynamics approximates the
  desired controlled evolution.  Although efficient, this scheme uses
  operations randomly sampled from an exponentially large set.  In the
  present work, we show that more efficient controlization schemes can
  be constructed with the help of orthogonal arrays for unknown
  2-local Hamiltonians. We conduct a detailed analysis of their
  performance and demonstrate the resulting improvements through
  numerical experiments. This construction can also be generalized to
  $k$-local Hamiltonians. Moreover, our controlization schemes based
  on orthogonal arrays can take advantage of the interaction graph's
  structure and be made more efficient.}

\section{Introduction}

The task of efficiently realizing controlled operations plays a
fundamental role in the design and execution of quantum
algorithms. Algorithmic primitives such as quantum phase estimation
\cite{nielsen00}, quantum amplitude estimation
\cite{brassard2002amplitude}, linear-combination-of-unitaries (LCU)
\cite{childs2012hamiltonian,kothari14}, and quantum singular value
transformation (QSVT) \cite{Low2019qubitization,gilyen2018singular}
require implementing a \emph {controlled} unitary $U$, an operation of
the form $\mathrm{ctrl}(U)=|0\>\<0|\otimes I + |1\>\<1|\otimes U$. In
applications related to quantum simulation, the unitary $U$ is often
the time-evolution under a Hamiltonian $H$ for time $t$, i.e.,
$U = e^{-iHt}$. Quantum phase estimation, for example, is used with
controlled time-evolution to determine spectral properties of
Hamiltonians. Block-encoding techniques are used to manipulate
controlled time-evolution into implementing functions of $H$ -- this
idea is foundational to quantum algorithms for preparing both ground
and Gibbs states of quantum many-body systems
\cite{poulin2008ground,poulin2009sampling,chowdhury2017gibbs,dong2022ground,silva2023fragmented,ding2024robust},
solving linear systems of equations
\cite{harrow2009linear,childs2017linear}, and many more
\cite{Low2019qubitization,gilyen2018singular}.

In most of these examples, the Hamiltonian $H$ is known
beforehand. Thus there is an explicit quantum circuit that realizes
the time-evolution and turning that into a controlled operation is
straightforward.  There exist problems, for instance, in quantum
meteorology \cite{pang2014metrology} and quantum learning
\cite{wiebe2014hamiltonianlearning,li2020hamiltoniantomo,anshu2020sampleefficient},
where the Hamiltonian may not be known a priori. However, we may still
wish to use the quantum algorithmic primitives mentioned above to
design efficient algorithms. This raises the question: ``Can we
\emph{controlize} an unknown Hamiltonian evolution? That is, given the
ability to perform $\exp{(-iHt)}$ as a black-box, can we implement the
operation $\mathrm{ctrl}\left(\exp{(-iHt)}\right)$?"

Controlization of unknown quantum dynamics has been studied in a
number of different settings
\cite{janzing2002quantum,nakayama2015energy,dong2021controlled,odake2024universal}. Janzing
\cite{janzing2002quantum} showed how a controlization protocol can be
used to perform phase estimation of unknown quantum Hamiltonians.
Closely related are also the problems of reversing and fast-forwarding
Hamiltonian dynamics
\cite{navascues2018reversing,trillo2020translating}. While our focus
here is on Hamiltonian dynamics, it should be noted that controlizing
arbitrary unitary operations has also been investigated extensively
\cite{araujo2014quantum,chiribella2016optimal,quintino2019reversing,dong2021controlled}. Controlization
has also been used as a primitive in algorithms for transforming
unknown Hamiltonian dynamics to implement functions of Hamiltonians
\cite{odake2024universal}.

The controlization protocol for unknown Hamiltonian dynamics in
Ref.~\cite{odake2024universal}, which is based on prior work in
Refs.~\cite{nakayama2015energy,dong2021controlled}, proceeds as
follows. It invokes the black-box to implement the Hamiltonian
dynamics $\exp(-i H \tau)\in \cU (\cH)$ for different times $\tau>0$
and intersperses them with different control operations from a finite
subset of $\cU(\C^2\otimes \cH)$ such that the resulting evolution
approximates the unitary
\begin{align}
    \mathrm{ctrl}\big(\exp(-i H t)\big) 
    &=
    |0\>\<0| \otimes I + |1\>\<1| \otimes \exp(-i H t) \in \cU(\C^2 \otimes \cH)
\end{align}
for some desired value of $t$.  More precisely, the unknown
Hamiltonian $H$ acts on $n$ qubits, that is, $\cH=(\C^2)^{\otimes n}$
and the control operations are controlled $n$-fold tensor products of
the Pauli matrices $\{I,X,Y,Z\}$. In particular, the operations are
randomly sampled from an exponentially large set.

The purpose of this work is to present efficient controlization
schemes for the situation where the quantum system consists of $n$
qudits, that is, $\mathcal{H}=(\C^d)^{\otimes n}$, and the unknown
system Hamiltonian $H\in\mathcal{L}(\mathcal{H})$ is assumed to be
$k$-local with $k$ being a constant. Our construction of
controlization schemes is based on the combinatorial concept of
orthogonal arrays \cite{hedayat1999orthogonal} and improves upon some
aspects of the controlization schemes presented in
\cite{odake2024universal}. It significantly reduces the size of the
set of operations used to intersperse the black-box Hamiltonian
dynamics.  The main idea underlying our construction is that there is
a close connection between controlization and decoupling schemes and
that the latter can be constructed from orthogonal arrays. We analyze
the performance of our OA-based schemes and show improved results over
existing methods for $k$-local Hamiltonians both theoretically and
through numerical experiments.

\section{Controlization of known Hamiltonian dynamics}

Before considering the case of unknown Hamiltonian dynamics, let us
illustrate how a decoupling scheme for a Hamiltonian $H$ for which the
terms are explicitly known can be naturally extended to a
controlization scheme.

For simplicity assume initially that $H = P$, where $P$ is a Pauli
operator. Choosing a Pauli $Q$ that anti-commutes with $P$ enables us
to decouple, i.e., to switch off the time evolution, since
\[
P + QPQ^{\dag} = P - P = 0
\]
and thus
\[
(e^{-iPt/2})(Qe^{-iPt/2}Q^{\dag}) =
(e^{-iPt/2})(e^{-iQPQ^{\dag}t/2}) =
(e^{-iPt/2})(e^{+iPt/2}) = I.
\]
Omitting $Q$ and $Q^{\dag}$ in the product above yields
\[
(e^{-iPt/2})(e^{-iPt/2}) =
e^{-iPt}
\]
This allows us to implement the conditional time evolution with $P$ as shown in the circuit below:

\begin{center}
\includegraphics[height=24pt]{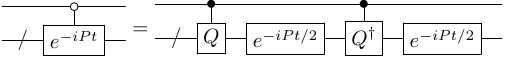}
\end{center}
Since we can always choose $Q$ to be a weight-one Pauli, the advantage
of this approach is that we replace a conditional multi-qubit Pauli
exponentiation gate by one that merely requires two conditional
single-qubit Pauli operations.

Time evolution for a known Hamiltonian $H = \sum_{i=1}^N \alpha_i H_i$
consisting of a weighted sum of Pauli operators can be implemented by
means of Trotterization or schemes such as qDRIFT
\cite{campbell2019random,kiss2023importancesampling}.  The latter is a
probabilistic simulation technique that often outperforms
deterministic simulation techniques. An important benefit of qDRIFT is
that it is independent of the number $N$ of Trotter terms but instead
depends on the absolute sum of Hamiltonian strengths $\alpha_i$.

A conditional time evolution is then obtained by splitting the time
evolution of each individual term in two pieces and conjugate one of
them by an operator that negates the term. Various circuit
simplifications can be made. For instance, if the Pauli terms $A$ and
$B$ for successive exponentiation terms $e^{iAt}$ and $e^{iBt}$, have
overlapping support we can always find a unit-weight Pauli $Q$ on that
support that anti-commutes with both terms. This allows us to write:

\begin{center}
\includegraphics[height=24pt]{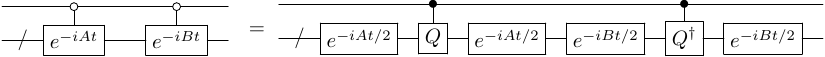}
\end{center}

The controlled $Q$ and $Q^{\dag}$ gates originally present between
$e^{-iAt/2}$ and $e^{-iBt/2}$ cancel and are therefore
omitted. Grouping the Pauli terms into commuting sets and
appropriately choosing the ordering of the terms allows additional
simplification along the same lines.

\section{Controlization of unknown Hamiltonian dynamics based on
  decoupling schemes}

The goal of dynamical decoupling is to effectively stop the
system\footnote{More generally, one can consider a subsystem.} from
evolving by interspersing the time evolution according to the system
Hamiltonian $H$ by suitable control operations. The works
\cite{lorenza99dynamical, zanardi00symmetrizing} were the first to
apply dynamical decoupling for purposes of quantum information
processing. We refer the reader to
\cite{dodd02universal,nielsen02universal,wocjan2002simulating,
  wocjan2002universal,wocjan2003computational} for an introduction to
a combinatorial approach for constructing Hamiltonian simulation
schemes based on average Hamiltonian theory. Decoupling and
time-reversal are two special cases of these schemes. We make use of
the combinatorial decoupling schemes presented in these works to
construct our controlization schemes.

In the current section we do not assume any particular structure of
the quantum system $\cH$ and its Hamiltonian $H\in\cL(\cH)$. We only
assume the Hamiltonian $H$ to be traceless throughout. We also assume
that the control operations in the control group $\cG$ -- a finite
subgroup of $\cU(\cH)$ -- can be implemented instantaneously (this is
often referred to as bang-bang control).\footnote{We refer the reader
  to the paper \cite{bookatz2014hamiltonian, bookatz2016improved} and
  the references therein for Hamiltonian simulation based on
  bounded-strength controls.}  In the next section we consider quantum
systems consisting of qudits and local Hamiltonians describing the
couplings among the qudits.

To illustrate the idea of Hamiltonian simulation based on average
Hamiltonian theory, let us consider a simple example that uses three
control operations $V_1$, $V_2$, $V_3\in\cG$ and invokes the black-box
to implement the dynamics for two times $\tau_1$ and $\tau_2$.  We can
rewrite the resulting unitary as follows:
\begin{align}
    &
    V_3 \, e^{-i H \tau_2} \, V_2 \, e^{-iH \tau_1} \, V_1 \\
    &=
    (V_3 V_ 2 V_1) \, (V_2 V_1)^\dagger \, e^{-iH \tau_2} \, (V_2 V_1) \, V_1^\dagger \, e^{-iH \tau_1} \, V_1 \\
    &=
    e^{-i U_2^\dagger H U_2 \tau_2} \, e^{-i U_1^\dagger H U_1 \tau_1}
\end{align}
The operators $U_1$ and $U_2$ are given by
\begin{align}
    U_1 &= V_1 \\
    U_2 &= V_2 V_1
\end{align}
assuming that the final control operation $V_3$ is chosen such that
$V_3 V_2 V_1=I$ (the latter means that the control sequence is
cyclic). This method of rewriting the resulting time evolution
generalizes to an arbitrary number of control operations.  For cyclic
control schemes we have
\begin{align}
    U_1 &= V_1 \\
    U_2 &= V_2 V_1 \\
        &\,\,\, \vdots \nonumber \\ 
    U_N &= V_N \cdots V_2 V_1 ,
\end{align}
where $V_{N+1} = (V_N \cdots V_2 V_1)^\dagger$.
Using the simple Trotter approximation, we obtain
\begin{equation}
    \prod_{j=1}^N e^{-i U_j^\dagger H U_j \tau_j} 
    =
    e^{-i U_N^\dagger H U_N \tau_N} \, \cdots \, 
    e^{-i U_2^\dagger H U_2 \tau_2} \, 
    e^{-i U_1^\dagger H U_1 \tau_1} \\
    \approx
    e^{-i \tilde{H}},
\end{equation}
where the average Hamiltonian $\tilde{H}$ is given by the sum
\begin{align}\label{eq:avg_Ham}
    \tilde{H} &= \tau_1 U_1 H U_1^\dagger + \tau_2 U_2 H U_2^\dagger + \ldots + \tau_N U_N H U_N^\dagger.
\end{align}

Thus, after the control cycle the resulting time evolution is
approximately as if the system had evolved under the average
Hamiltonian $\tilde{H}$.  This is referred to as average Hamiltonian
theory.

A better approximation can be achieved according to the formula:
\begin{align}
    \prod_{j=N}^1 e^{-i U_j^\dagger H U_j \tau_j/2} 
    \prod_{j=1}^N e^{-i U_j^\dagger H U_j \tau_j/2} 
\end{align}
Take note that the two products above multiply the exponentials
$e^{-i U_j^\dagger H u_j \tau_j/2}$ in opposite orders and use times
$\tau_j/2$.  This formula is a special case of the Lie-Trotter-Suzuki
formulas \cite{SUZ1990a, SUZ1991a}. Observe that these higher-order
formulas do not increase the overall time of the control cycle, but
only increase the number of control operations.

The average Hamiltonian in eq.~(\ref{eq:avg_Ham}) serves as the
starting point for these higher-order formulas. This is why it is
important to express the desired target Hamiltonian $\tilde{H}$ as a
weighted sum containing a small number $N$ of conjugates
$U_1^\dagger H U_1,\ldots,U_N^\dagger H U_N$ and having a small
overall time $\sum_{j=1}^N \tau_j$.

\begin{definition}[Simulation, decoupling, and time reversal schemes]
  Let $H,\tilde{H}\in\cL(\cH)$ be two Hamiltonians. We say that a
  scheme $\cS=(\tau_1,U_1;\ldots;\tau_N,U_N)$ is a simulation scheme
  for the system Hamiltonian $H$ and the target Hamiltonian
  $\tilde{H}$ if
\begin{align}
    \cS(H) 
    &= \sum_{j=1}^N \tau_j U_j H U_j^\dagger = \tilde{H}.
\end{align}
There are two important special cases. For $\tilde{H}=0_\cH$, where
$0_\cH\in\cL(\cH)$ denotes the zero operator, we call the scheme a
\emph{decoupling} scheme whenever $\sum_i \tau_i = 1$ and use $\cD$ to
denote it. For $\tilde{H}=-H$, we call the scheme a \emph{time
  reversal} scheme and use $\cR$ to denote it.
\end{definition}

\begin{remark}[Decoupling $\Rightarrow$ time reversal] 
  Before discussing in detail how decoupling schemes can be extended
  to controlization schemes, let us briefly describe the connection to
  time reversal schemes.  Let
  $\cD=(\tau_1, U_1; \tau_2, U_2; \ldots; \tau, U_N)$ be a decoupling
  scheme. We may assume, without loss of generality, that
  $U_1=I_\cH$. Then, we have
\begin{align}
    \tau_1 H + \tau_2 U_2 H U_2^\dagger + \ldots \tau_N U_N H U_N^\dagger 
    &= 0,
\end{align}
which is equivalent to
\begin{align}
    \frac{\tau_2}{\tau_1} \, U_2 H U_2^\dagger + \ldots + \frac{\tau_N}{\tau_1} \ U_N H U_N^\dagger 
    &= -H.
\end{align}
Thus, the decoupling scheme
$\cD=(\tau_1, U_1 = I; \tau_2, U_2; \ldots; \tau_N, U_N)$ gives rise
to the time reversal scheme
$\cR=(\tau_2 / \tau_1, U_2; \ldots; \tau_N / \tau_N, U_N)$.  Observe
that the time needed to realize time reversal is $(1-\tau_1)/\tau_1$
which will be greater than $1=\tau_1 + \ldots + \tau_N$ in contrast to
the situation for decoupling. This slow-down for time-reversal is
discussed in \cite{janzing2002complexity} and more generally for
Hamiltonian simulation in \cite{wocjan2002simulating}.
\end{remark}

Before we state the theorem connecting controlization and decoupling,
we need to introduce the following definition.  Let $M\in\cL(\cH)$ be
an arbitrary operator. Then, $\Lambda(M)$ is the controlled operator
defined to be
\begin{align}
    \Lambda(M) &= |0\>\<0| \otimes M + |1\>\<1| \otimes I_\cH \in \cL(\C^2 \otimes \cH),
\end{align}
where $I_\cH\in\cL(\cH)$ denotes the identity operator.

\begin{theorem}[Decoupling $\Rightarrow$ controlization]\label{thm:decoupling_ctrl}
  Let $\cD=(\tau_1,U_1;\ldots;\tau_N,U_N)$ be a decoupling scheme for
  some Hamiltonian $H\in\cL(\cH)$. Then, the scheme
  $\cC=\Lambda(\cD)=(\tau_1,\Lambda(U_1);\ldots;\tau_N,\Lambda(U_N))$
  is a simulation scheme for the pair
\begin{align}
    I_2 \otimes H \mbox{ and } |1\>\<1| \otimes H
\end{align}
of Hamiltonians in $\cL(\C^2\otimes \cH)$.  For this reason, $\cC$ is
a controlization scheme that makes it possible to approximately
implement the controlled time evolution according to $H$, that is,
\begin{align}\label{eq:ctrl_ham}
    |0\>\<0| \otimes I_\cH + |1\>\<1| \otimes \exp(-i H t).
\end{align}
\end{theorem}

\begin{proof}
  We adjoin a single qubit $\C^2$ to the quantum system $\cH$. The
  overall Hamiltonian of the joint system $\C^2\otimes\cH$ is of the
  form
\begin{align}
    I_2 \otimes H,
\end{align}
that is, there is no coupling between $\C^2$ and $\cH$ and the control
qubit $\C^2$ does not have any non-trivial internal dynamics. We have
\begin{align}
    \Lambda(U_j) &= |0\>\<0| \otimes U_j + |1\>\<1| \otimes I_\cH
\end{align}
as the control operations in $\cC$ for $k=1,\ldots,N$. We obtain
\begin{align}
    \cC(I_2\otimes H) 
    &=
    \sum_{j=1}^N \tau_j \Lambda(U_j) (I_2 \otimes H) \Lambda(U_j)^\dagger \\
    &=
    |0\>\<0| \otimes \sum_{j=1}^N \tau_j U_j H U_j^\dagger +
    |1\>\<1| \otimes \sum_{j=1}^N \tau_j H \\
    &=
    |0\>\<0| \otimes \cD(H) + |1\>\<1| \otimes H \\
    &=
    |1\>\<1| \otimes H,
\end{align}
where we used the fact that $\cD(H)=0_\cH$ since $\cD$ is a decoupling
scheme for $H$ and $\sum_{j=1}^N\tau_j=1$ by convention for decoupling
schemes.  Finally, observe that time evolution according to the target
Hamiltonian $|1\>\<1|\otimes H$ gives rise to the desired controlled
time evolution in eq.~(\ref{eq:ctrl_ham}).
\end{proof}

\begin{remark}
  The idea of converting a decoupling scheme into a controlization
  scheme is sketched but not pursued in \cite{janzing2002quantum}
  because a different model is considered in that work. Let us briefly
  explain the main differences. Observe that in our theorem the
  operations $U_j$ of the decoupling scheme become controlled
  operations of the form
  $\Lambda(U_j)=|0\>\<0|\otimes U_j + |1\>\<1| \otimes I$ of the
  controlization scheme. Such controlled operations are specifically
  disallowed in that work. For instance, when $\cH=(\C^2)^{\otimes n}$
  and the control unitaries $U_j$ are tensor products of single qubit
  gates, the implementation of $\Lambda(U_j)$ would require two-qubit
  gates.  While bang-bang control of single qubit operations is often
  a valid assumption, bang-bang control of two-qubit operations is
  more problematic.

  Therefore, that work considers the following model.\footnote{Both to
    simplify the presentation and to facilitate the comparison with
    our model, we describe here the model of \cite{janzing2002quantum}
    as having a \emph{single} qubit as control, whereas that work
    considers more generally a \emph{multi-qubit} register.} The joint
  quantum system is $\C^2\otimes \cH$, where $\C^2$ is the ancilla
  qubit and $\cH$ is an $n$-qubit system. The former is called the
  control qubit and the latter the target register. The starting
  Hamiltonian $H_{\mathrm{init}}$ is a pair-interaction Hamiltonian of
  the form
\begin{align}
    H_{\mathrm{init}} &= H_{\C^2} \otimes I_\cH + H_{\C^2,\cH} + I_2 \otimes H.
\end{align}
The term $H_{\C^2}$ is the Hamiltonian of the control qubit and the term $H_{\C^2,\cH}$ specifies the pair-interactions\footnote{The pair-interactions between the control qubit and the qubits of the target register need to be non-trivial to enable one to construct a simulation scheme based on average Hamiltonian theory.} between the control qubit and the qubits of the target register. Both terms are assumed to be known. The term $H$ is the Hamiltonian of the target register.  It is an unknown pair-interaction Hamiltonian, and the task is to controlize it. More precisely, that work constructs a simulation scheme $\cC$ using only single qubit operations such that
\begin{align}
    \cS(H_{\mathrm{init}}) &= \sigma_z \otimes H.
\end{align}
Note that having $\sigma_z\otimes H$ is equivalent to having $|1\>\<1|\otimes H$.
\end{remark}

\section{Decoupling schemes based on orthogonal arrays}

The previous section established the connection between controlization
and decoupling schemes without making any assumptions about the
quantum system $\cH$ and the Hamiltonian $H\in\cL(H)$.

In this section we consider $k$-local Hamiltonians acting on $n$
qudits, that is, $\cH=(\C^d)^{\otimes n}$. We review some known
constructions of decoupling schemes based on the combinatorial concept
of orthogonal arrays, which then give rise to controlization schemes
via Theorem~\ref{thm:decoupling_ctrl}. Orthogonal arrays, which were
first used for decoupling two-local qubit Hamiltonians
in~\cite{stollsteimer2001suppression}, are defined as
follows~\cite{hedayat1999orthogonal}:

\begin{definition}[Orthogonal arrays]
  An $N\times n$ array $M$ with entries from a finite alphabet $\cA$
  is an $OA(N,n,s,t)$ orthogonal array with $s=|\cA|$ levels and
  strength $t$ if and only if each $N\times t$ subarray contains each
  $t$-tuple of elements of $\cA$ as a row exactly $\lambda=N/s^t$
  times.
\end{definition}

The challenge is to construct an orthogonal array with the smallest
possible $N$ for given parameters $n$, $s$, and $t$. The parameters of
the orthogonal array are related to parameters of the decoupling
schemes as follows: (i) $N$ is the number of time steps, (ii) $n$ is
the number of qudits, (iii) $s=d^2$, where $d$ is the dimension of the
qudits, and (iv) $t=k$, where $k$ is the locality of the qudit
Hamiltonian.

We consider a quantum system $\cH=(\C^d)^{\otimes n}$ that consists of
$n$ interacting $d$-dimensional qudits and assume that the system
Hamiltonian $H\in\cL(\cH)$ is an arbitrary $k$-local operator, that
is,
\begin{align}
    H &= \sum_{(i_1,i_2,\ldots,i_k)} H_{(i_1,i_2,\ldots,i_k)},
\end{align}
where the $k$-tuples $(i_1,i_2,\ldots,i_k)$ in the sum run over all
tuples with $1\le i_1 < i_2 < \ldots < i_k\le n$ and the corresponding
terms $H_{(i_1,i_2,\ldots,i_k)}$ are (traceless) Hermitian operators
acting only the $k$ qudits specified by the entries of
$(i_1,i_2,\ldots,i_k)$.

Let us consider a simple example of a $3$-local Hamiltonian $H$ acting
on $4$ qubits, that is, $d=2$, $n=4$, and $k=3$. Let $X$, $Y$, and $Z$
denote the Pauli matrices. The example Hamiltonian
$H\in\cL\big((\C^2)^{\otimes 4}\big)$ is
\begin{align}
    H 
    &= 
    X \otimes X \otimes I \otimes X + 
    Y \otimes I \otimes I \otimes Y +
    \frac{1}{2} \big(
    Z \otimes I \otimes Z \otimes Z +
    X \otimes I \otimes X \otimes X
    \big) 
\end{align}
and can be expressed as 
\begin{align}
    H 
    &=
    H_{(1,2,4)} + H_{(1,3,4)},
\end{align}
where the terms $H_{(1,2,4)}$ and $H_{(1,3,4)}$ are given by
\begin{align}
    H_{(1,2,4)} 
    &= 
    X \otimes X \otimes I \otimes X + 
    Y \otimes I \otimes I \otimes Y \\
    H_{(1,3,4)} 
    &=
    \frac{1}{2} \big(
    Z \otimes I \otimes Z \otimes Z +
    X \otimes I \otimes X \otimes X
    \big). 
\end{align}
Note that the $2$-local term $Y \otimes I \otimes I\otimes Y$ could be
included in $H_{(1,3,4)}$ instead of $H_{(1,2,4)}$.

\begin{theorem}[OA $\Rightarrow$ decoupling scheme]
  Let $\cA$ be the finite alphabet $\{1,\ldots,d^2\}$. Then, any
  orthogonal array with parameters $OA(N,n,d^2,k)$ over $\cA$ can be
  used to define a decoupling scheme $\cD$ that annihilates any
  $k$-local Hamiltonian on $n$ qudits. The number of local operations
  used in this scheme is given by $N$.
\end{theorem}

The proof for $2$-local ($k=2$) $n$ qudit Hamiltonians is given in
\cite[Theorem 1]{wocjan2002simulating}. The general case of qudit
Hamiltonians with locality $k\ge 2$ is established in \cite[Theorem
8]{roetteler2006equivalence}. We give a brief summary of the proof
below and refer the reader to the references for more details.

The unitaries $U_j$ for $j=1,\ldots,N$ in the OA-based decoupling
schemes are local operations of the form $U_j\in\cG^{\otimes n}$,
where $\cG$ is a finite subgroup of the unitary group $\cU(d)$.  The
times $\tau_j$ are all equal to $\frac{1}{N}$.  We now explain the
correspondence between the entries of orthogonal array and the control
operations.  Define the $d^2$ many generalized Pauli matrices
\begin{align}
    X^a Z^b \in \cG
\end{align}
for $a,b\in\{0,\ldots,d-1\}$, where
\begin{align}
    X &= \sum_{x=0}^{d-1} |x + 1 \mbox{ mod } d\>\<x| \\
    Z &= \mathrm{diag}(\omega^0,\omega^1\ldots,\omega^{d-1}),
\end{align}
and $\omega=\exp(2\pi i/d)$ is a $d$th root of unity.  Denote these
generalized Pauli matrices by $P_1,\ldots,P_{d^2}$. This collection is
a special case of a so-called nice error basis for $\C^d$
\cite{knill1996nonbinary, klappenecker2002beyond}.  For the special
case $d=2$, this yields the Pauli matrices.

It can be shown that all $k$-fold tensor products of generalized Pauli
matrices give rise to the depolarizing channel $\Phi$ on
$(\C^d)^{\otimes k}$. More precisely, the depolarizing channel is
given by
\begin{align}
    \Phi(\bullet) 
    &= 
    \frac{1}{|\cA^k|} \sum_{(i_1,i_2,\ldots,i_k)\in\cA^k}
    (P_{i_1} \otimes P_{i_2} \otimes \cdots \otimes P_{i_k}) 
    \bullet 
    (P_{i_1} \otimes P_{i_2} \otimes \cdots \otimes P_{i_k})^\dagger,
\end{align}
where $\cA^k$ denotes the $k$-fold direct product of
$\cA=\{1,2,\ldots,d^2\}$. In particular, any traceless Hamiltonian
$h\in\cL\big((\C^d)^{\otimes k})$ is annihilated by $\Phi$, that is,
$\Phi(h)=0$.

This annihilation property can be ``lifted'' to $k$-local $n$-qudits
Hamiltonians with help of the orthogonal array as follows. The $j$th
control operation $U_j$ is the $n$-fold tensor product of generalized
Pauli matrices defined to be
\begin{align}
    U_j 
    &= 
    P_{m_{j1}} \otimes P_{m_{j2}} \otimes \cdots \otimes P_{m_{jn}}\in\cG^{\otimes n},
\end{align}
where the entries $(m_{j1}, m_{j2},\ldots,m_{jn})\in\cA^n$ correspond
to the $j$th row of the orthogonal array.  It can be shown that the
scheme simultaneously defines a depolarizing channel on all
$k$-subsets of the $n$ qudits. This is due to the defining property of
orthogonal arrays: each $k$-tuple in $\cA^k$ appears exactly $\lambda$
times for any $k$-subset of columns of the orthogonal array.

\begin{figure}
\begin{center}
\begin{tabular}{|cccccccccccccccc|} 
\hline
1& 1& 1& 1& 2& 2& 2& 2& 3& 3& 3& 3& 4& 4& 4& 4 \\
1& 2& 3& 4& 1& 2& 3& 4& 1& 2& 3& 4& 1& 2& 3& 4 \\ 
1& 2& 3& 4& 4& 3& 2& 1& 2& 1& 4& 3& 3& 4& 1& 2 \\ 
1& 2& 3& 4& 2& 1& 4& 3& 3& 4& 1& 2& 4& 3& 2& 1 \\
1& 2& 3& 4& 3& 4& 1& 2& 4& 3& 2& 1& 2& 1& 4& 3 \\ \hline
\end{tabular}
\end{center}
\caption{Orthogonal array $OA(16, 5)$ with multiplicity $1$. The table shows the transpose of the OA.}
\label{tab:small_OA}
\end{figure}

The construction of a decoupling scheme with the above approach
requires an orthogonal array of appropriate size. For qubits, we
consider the alphabet $\cA=\{1, 2, 3, 4\}$ corresponding to the Pauli
matrices $\{I, X, Y, Z\}$, which gives $s=\vert\cA\vert =
4$. Figure~\ref{tab:small_OA} shows an example orthogonal array, which
can be used to decouple two-local ($k=2$) Hamiltonian on up to $n=5$
qubits with $N=16$ control operations (columns of the orthogonal array
can always be removed whenever $n < 5$).

Explicit constructions of orthogonal arrays are given
in~\cite{hedayat1999orthogonal}. The following two constructions yield
orthogonal arrays of strength $2$, which can be used to decouple
two-local Hamiltonians:

\begin{theorem}[{\cite[Theorem 3.20]{hedayat1999orthogonal}}]\label{thm:3.20}
  If $s$ is a prime power then an
  $OA(s^{\ell}, (s^{\ell}-1)/(s-1),s,2)$ exists whenever $\ell\geq 2$.
\end{theorem}

\begin{theorem}[{\cite[Theorem 6.40]{hedayat1999orthogonal}}]\label{thm:6.40}
  If $s$ is a power of a prime and $\ell\geq 2$, then an orthogonal
  array $OA(2s^{\ell},2(s^{\ell}-1)/(s-1)-1,s,2)$ can be obtained by
  using difference schemes.
\end{theorem}

\noindent Applying these theorems in our setting gives orthogonal
arrays of the size shown in Table~\ref{table:OAs}. Note that most, if
not all, constructions in~\cite{hedayat1999orthogonal} require
$s=d^2$, and therefore $d$ itself, to be a prime power.

Tables of select orthogonal arrays of strength $2$ and higher are
available; see for instance~\cite{beth99design, colbourn2006handbook,
  hedayat1999orthogonal}. There is a close relationship between error
correcting codes and orthogonal arrays \cite[Theorem
4.6]{hedayat1999orthogonal}. Certain BCH-codes were used in
\cite[Theorem 6]{bookatz2016improved} to construct orthogonal arrays
of strength $t\ge 2$, which can be used to decouple $k$-local
$n$-qudits Hamiltonians.

\begin{theorem}\label{thm:bch_oa}
  For any $t\ge 2$, $n\ge (t-1)^2$, and $s=d^2$ with $d\ge 2$ a prime
  power, there exists an $OA(N,n,s,t)$ whose length $N$ scales as
  $N=O(n^{t-1})$. That is, there exists an OA decoupling scheme to
  decouple $k$-local Hamiltonians on $n$ qudits that uses
  $N=O(n^{k-1})$ time steps.
\end{theorem}

In physical systems, the locality $k$ is generally a small fixed
number. The above statement focuses on the asymptotic cost for fixed
locality $k$ as the number $n$ of qudits grows. It should be noted
that, depending on the particular choice of parameters further
improvements are possible.

\begin{table}
\centering
\begin{tabular}{|l|l|l|}
\hline
$\gamma=1$ & $OA(16,5,4,2)$ & Follows from Theorem~\ref{thm:3.20}\\
$\gamma=2$ &$OA(32,9,4,2)$ & Follows from Theorem~\ref{thm:6.40}\\
$\gamma=4$ &$OA(64,21,4,2)$ & Follows from Theorem~\ref{thm:3.20}\\
$\gamma=8$ &$OA(128,41,4,2)$ & Follows from Theorem~\ref{thm:6.40}\\
$\gamma=16$ &$OA(256,85,4,2$ & Follows from Theorem~\ref{thm:3.20}\\
$\gamma=32$ &$OA(512,169,4,2)$ & Follows from Theorem~\ref{thm:6.40}\\
\hline
\end{tabular}
\caption{Example orthogonal array parameters for used in Pauli ($s=4$)
  decoupling schemes for two-local Hamiltonians ($k=2$). For instance,
  based on the second entry, an arbitrary pair-interaction Hamiltonian
  on $n=9$ qubits can be decoupled using $N=32$ control operations,
  where each control operation is tensor product of Pauli
  operators.}\label{table:OAs}
\end{table}

\subsection{Two-local Hamiltonians based on coupling topology}

We now consider the situation of two-local $n$-qudit Hamiltonians when
the qudits are not all coupled to each other, that is, the interaction
graph of the Hamiltonian is assumed to be a known non-complete graph.
For instance, this is the case when the qudits are arranged on a
square lattice and there is a coupling only between nearest neighbors
(only one site away either horizontally or vertically).

We color the vertices of this graph by assigning each vertex one of a
number of different colors. We need the coloring to be proper meaning
that connected vertices always receive different colors.  The
chromatic number is $\chi$ is the minimum number of colors required to
achieve a proper coloring. The chromatic number $\chi$ of a partially
coupled qudit Hamiltonian can be significantly smaller than $n$, which
is the chromatic number of the complete graph. For the square lattice
example, the chromatic number is $2$ independent of the number of
qudits.

This observation enables us to construct more efficient decoupling
schemes since there are no constraints on the control sequences
between qudits with the same color. Thus, it suffices to construct
decoupling scheme of only a fully coupled $\chi$-qudit (and not a
fully coupled $n$-qudit) system, and apply identical control
operations to qudits of the same color. To summarize, the chromatic
number becomes the effective number of qudits when selecting an
orthogonal array. This insight has already been used to obtain
improved decoupling schemes. For instance, the early
work~\cite{jones1999efficient} and the recent
work~\cite{brown2024efficient} applied it to NMR quantum computing and
superconducting qubit devices, respectively.

\subsection{Example of orthogonal array performance}

Applying a decoupling scheme to a time-evolved Hamiltonian operator
$U(\delta t) = e^{-iH\delta t}$ on an arbitrary initial state
$\ket{\psi}$ should ideally leave the state untouched. As an
illustration of the practical performance of orthogonal arrays, we
consider an 8-qubit scenario where $\ket{\psi}$ is sampled from the
Haar distribution and the orthogonal array is chosen as a restriction
of $OA(32,9,4,2)$~(see \cite{Sloane_OA_Library}).  In order to
decouple $H$ we use a first-order Trotter scheme on the 32 components
$P_i H P_i$, where $P_i$ are the Pauli terms in the orthogonal array,
and the operator is implemented as $P_i U(\delta t) P_i^{\dag}$. A
second-order Trotter-Suzuki scheme can be defined similarly, resulting
in a scheme consisting of 64 terms (when not merging of the central
two term in the scheme).  We consider blocks of 64 terms as a single
unit, which amounts to two Trotter steps in the first-order scheme,
and a single step in the second-order scheme. The order of the Pauli
terms in the orthogonal array in each block can either be fixed or
randomly permuted, resulting in deterministic and randomized schemes,
respectively. When considering the average performance over these
schemes, we apply a single instance of the scheme to $\ket{\phi}$,
resulting in a final state $\ket{\phi^{(i)}}$. Repeating this over $r$
scheme instances, we characterize the performance by the trace
distance
\[
\frac{1}{2}
\mbox{Tr}\left|\ket{\phi}\bra{\phi} - \frac{1}{r}\sum_{i=1}^r \ket{\phi^{(i)}}\bra{\phi^{(i)}}\right|.
\]
For the randomized schemes we take $r=100$. For the deterministic
scheme all instances are identical, which allows us to take $r=1$.
The resulting trace distance is shown in
Figure~\ref{Fig:OA_performance} for different numbers of blocks,
corresponding to different numbers of Trotter steps. Each
light-colored line gives the trace distance for a randomly sampled,
but fixed initial state $\ket{\phi}$. The solid line is then taken as
the average trace distance over ten random initial states.  The
results show a clear advantage of the second-order scheme over the
first-order scheme. Further improvements are obtained in the
randomized scheme where a random permutation of the orthogonal array
terms is done for each Trotter step. For comparison, we also include
the performance of a qDRIFT-style approach to decoupling based on
sampling elements from the entire Pauli group. The trace distance for
this approach seems to have a slower start compared to that of the
orthogonal-array based approaches. One possible reason for this is
that the Pauli group ($4^8$ elements in this case) is much larger than
the number of elements in the orthogonal array, which means that it is
not possible to apply each element in the given number of blocks. As a
comparison, we therefore also applied the qDRIFT-style approach on the
much smaller set of Paulis corresponding to the elements in the
orthogonal array. The resulting lines in the plots exactly overlap
with those of the full Pauli group. We therefore suspect that the
difference between qDRIFT and the orthogonal-array based approach is
because the former does not sample the terms equally often, thereby
preventing exact cancellation of the terms. As the number of blocks
increases the relative imbalance of the terms in the qDRIFT-style
approach eventually becomes less pronounced, leading to a convergence
speed that seems to match that of the first-order Trotterization with
orthogonal arrays.

\begin{figure}
\centering
\begin{tabular}{cc}
\includegraphics[width=0.47\textwidth]{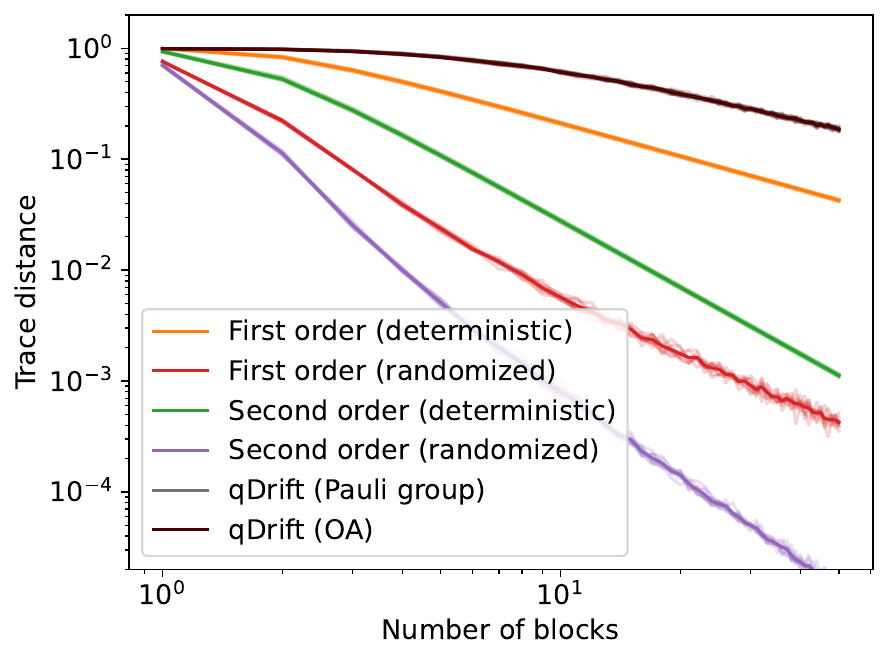}&
\includegraphics[width=0.47\textwidth]{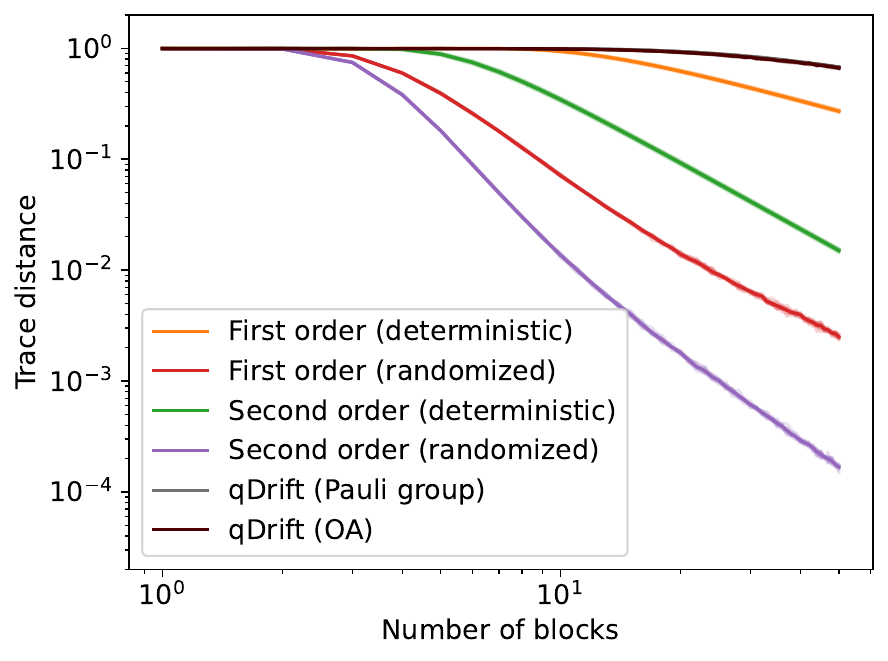}\\
({\bf{a}}) & ({\bf{b}})
\end{tabular}
\caption{Application of various decoupling schemes on randomly sampled
  initial states $\ket{\phi}$. The faint lines give the trace distance
  between the initial state, and the state obtained following
  application of the decoupled Hamiltonian time evolution operator
  $U$. The solid lines give the average of this distance over ten
  random initial states. Each block consists of 64 application of $U$
  flanked with relevant Pauli terms. A larger number of blocks
  corresponds to smaller time steps in the Trotterized
  implementation. We consider ({\bf{a}}) a sparse 8-qubit Hamiltonian
  with 40 random two-local Pauli terms, and ({\bf{b}}) an 8-qubit
  Hamiltonian with all one- and two-local Pauli terms, with
  coefficients chosen i.i.d.~uniformly at random over the interval
  $[0,2\pi)$.}\label{Fig:OA_performance}
\end{figure}

\section{Improved controlization}

We now discuss how the OA-based constructions can give improved
controlization schemes. We first bound the number of controlled-gates
needed to implement a controlized time-evolution assuming the
existence of an orthogonal array with given parameters. Then we use
the bounds from Theorems~\ref{thm:3.20} and \ref{thm:6.40} to
determine the number of controlled operations needed to controlize
$2$-local Hamiltonians.

\begin{theorem}[OAs and first-order Trotter approximation]\label{thm:OA1storder}
  Let $H$ be an arbitrary $k$-local Hamiltonian acting on $n$
  qudits. Assume there exists an $OA(N,n, d^2, k)$. We can realize a
  unitary $\tilde{U}$ such that
    \begin{align}
        \| \tilde{U} - \mathrm{ctrl}(e^{-i H t}) \| \le \varepsilon
    \end{align}
    using 
    \begin{align}
        O\left(\frac{N \, t^2 \, \|H\|^2}{\varepsilon}\right)
    \end{align}
    many controlled generalized Paulis.
\end{theorem}
\begin{proof}
    We rely on the following upper bound on the Trotter error.
Let $H=\sum_{j=1}^N H_j$ and
\begin{align}
    U_\mathrm{Trotter} 
    &=
    \left(
        e^{-i H_N t/r} \cdots e^{-i H_1 t/r}
    \right)^r.
\end{align}
Standard error bounds for product formula give that 
\begin{align}
    \| e^{-i H t} - U_\mathrm{Trotter} \| 
    &=
    O\left(\frac{t^2}{r} \left(\sum_{k=1}^N \| H_j \|\right)^2 \right),
\end{align}
where $\|\cdot\|$ is the spectral norm (see e.g., Lemma 6 in
\cite{childs2021theory}).

For decoupling/controlization, we have
\begin{align}
    H_j = \frac{1}{N} U_j H U_j^\dagger
\end{align}
for some unitaries $U_1,\ldots,U_N$ and $H=0$. Thus we obtain the upper bound 
\begin{align}
    \| I - U_\mathrm{Trotter} \| 
    &=
   O\left( \frac{t^2}{r} \, N^2 \, \frac{1}{N^2} \, \| H \|^2\right)
    \le O\left(\frac{t^2 \, \| H \|^2}{r}\right).
\end{align}
If we want to make the Trotter error less than $\varepsilon$, then we
have to choose
\begin{align}
    r = \Theta\left(\frac{t^2 \, \|H\|^2}{ \varepsilon}\right).
\end{align}
The total number of control operations required is $r N$; $N$ control
operations in each of the $r$ Trotter subintervals.
\end{proof}

\begin{theorem}[OAs and second-order Trotter approximation]\label{thm:OA2ndorder}
  Let $H$ be an arbitrary $k$-local Hamiltonian acting on $n$
  qudits. Assume there exists an $OA(N,n, d^2, k)$. We can realize a
  unitary $\tilde{U}$ such that
    \begin{align}
        \| \tilde{U} - \mathrm{ctrl}(e^{-i H t}) \| \le \varepsilon
    \end{align}
    using
    \begin{align}
        O\left(\frac{Nt^{3/2} \, \| H \|^{3/2}}{\varepsilon^{1/2}}\right).
    \end{align}
    many controlled generalized Paulis.
\end{theorem}
\begin{proof}
  Suppose $H=\sum_{j=1}^N H_j$ and consider the second-order Trotter
  approximation to $e^{-iHt}$
    \begin{align}
        U_\mathrm{Trotter} 
        &=
        \left(
            e^{-i H_1 \tfrac{t}{2r}} \cdots e^{-i H_N \tfrac{t}{2r}}
        \cdot
            e^{-i H_N \tfrac{t}{2r}} \cdots e^{-i H_1 \tfrac{t}{2r}}
        \right)^r.
    \end{align}
    Then, from Lemma 6 in \cite{childs2021theory}, we have
    \begin{align}
        \| e^{-i H t} - U_\mathrm{Trotter} \| 
        &= r \times O \left(\left(\frac{t}{r} \sum_{j=1}^N\|H_j\|\right)^3 \right) = O \left(\frac{t^3}{r^2}\left( \sum_{j=1}^N\|H_j\|\right)^3 \right).
    \end{align}
    For decoupling and controlization, the $H_j$'s
    are of the form $\frac{1}{N}U_jHU_j^\dagger$ for some $H$ and unitaries $U_1,\dots,U_N$, and we obtain
    \begin{align}
        \| e^{-i H t} - U_\mathrm{Trotter} \| = O \left(\frac{t^3}{r^2}\|H\|^3 \right).
    \end{align}
    To make the second-order Trotter error less than $\varepsilon$, we have to choose (times some constant)
    \begin{align}
        r = \Theta\left(\frac{t^{3/2} \, \| H \|^{3/2}}{\varepsilon^{1/2}}\right).
    \end{align}
    The total number of controlled operations is then $2 r N$.
\end{proof}

The bounds in Theorems \ref{thm:OA1storder} and \ref{thm:OA2ndorder}
are stated in terms of the parameter $N$ of the orthogonal array
$OA(N,n, d^2, k)$. For 2-local Hamiltonians on $n$ qudits, we now
appeal to Theorem \ref{thm:6.40} and determine how $N$ scales with
$n$. Setting $s=d^2$ in Theorem \ref{thm:6.40}, we need to choose
$\ell$ such that
\begin{align}
    2\tfrac{d^{2 \ell} - 1}{d^2 - 1} \ge n 
\end{align}
in order to decouple or controlize $n$-qudit Hamiltonians. It suffices to choose $\ell$ satisfying
\begin{align}
    d^{2 \ell} \ge n d^2-1
\end{align}
and thus the total number of controlled operations is
\begin{align}
    N &= 2d^{2\ell} = O(n).
\end{align}
Using this with Theorems \ref{thm:OA1storder} and \ref{thm:OA2ndorder}
we arrive at the following statement about controlizing $2$-local
Hamiltonians.
\begin{theorem}[Controlizing $2$-local Hamiltonians]
  Let $H$ be an arbitrary $2$-local Hamiltonian acting on $n$
  qudits. We can realize a unitary $\tilde{U}$ such that
    \begin{align}
        \| \tilde{U} - \mathrm{ctrl}(e^{-i H t}) \| \le \varepsilon
    \end{align}
    using 
    \begin{align}
        O\left(\frac{n \, t^2 \, \|H\|^2}{\varepsilon}\right)
    \end{align}
    controlized generalized Paulis with a first-order Trotter
    formula. Moreover, the number of controlled generalized Paulis can
    be improved to
    \begin{align}
        O\left(\frac{nt^{3/2} \, \| H \|^{3/2}}{\varepsilon^{1/2}}\right)
    \end{align}
    by applying a second-order Trotter formula.
\end{theorem}

A key distinguishing feature of the controlization schemes we propose
here is that they can be implemented deterministically, unlike those
in Ref.~\cite{odake2024universal}. However, we require that the
Hamiltonian be $k$-local, whereas the controlization scheme in
Ref.~\cite{odake2024universal} works for any $n$-qubit Hamiltonian and
is therefore more general.

The generality of Ref.~\cite{odake2024universal} is because they use a
decoupling protocol that averages over conjugations with all tensor
products of Pauli on $n$-qubits. But, as a result, the target
Hamiltonian $\ket{1}\bra{1}\otimes H$ generating the
controlled-evolution ends up being a sum of $O(4^n)$ terms. This means
that the controlization scheme cannot be implemented efficiently by
usual Trotter-Suzuki product formula, as we would need to implement
$4^n$ terms for each Trotter step. Ref.~\cite{odake2024universal}
overcomes this by using the qDRIFT protocol by
Campbell~\cite{campbell2019random}. The qDRIFT algorithm samples terms
from the Hamiltonian with probabilities according to their weights in
the decomposition and implements a randomized product-formula which
converges to the target evolution. The overall cost scales
polynomially with the sum of the weights of the individual terms
instead of the total number of terms. Since the total weight is $O(1)$
in the controlization scheme of~\cite{odake2024universal}, qDRIFT
enables an efficient implementation.

In contrast, the OA-based controlization schemes here can give a
target Hamiltonian with fewer terms we have access to more structural
information about $H$. For an unknown $n$-qubit Hamiltonian $H$ that
is $k$-local, one has to find an $\mathrm{OA}_\lambda(N, n, 4, k)$
with $N=\lambda 4^k$ such that the multiplicity $\lambda$ is as small
as possible. As discussed in Theorem~\ref{thm:bch_oa}, there exist
$OA(N, n, 4, k)$ such that $N=O(n^{k-1})$. In addition to relying on
such asymptotic results, one should always search the literature or
use a software package (e.g.,~\cite{eendebak2019}) to identify an OA
with the smallest possible $N$ for the concrete application at hand.
The OA-based decoupling scheme then gives a decomposition of the
effective Hamiltonian into $N$ terms. With this, a simple first-order
Trotter formula can implement controlization using $N$ operations to
intersperse the unknown Hamiltonian time evolution. Finally, let us
remark that $n$-qubit $2$-local Hamiltonians with an interaction graph
of degree at most $\Delta$, can always be decoupled and controlized
using an orthogonal array with $N=O(\Delta)$.

\section{Some future research directions}

Recall the correspondence between the $U_j$ and $V_j$ matrices for
decoupling schemes.
\begin{align}
    V_1     &= I \\
    V_j     &= U_j U_{j-1}^\dagger \mbox{ for } j=2,\ldots,N \\
    V_{N+1} &= U_N^\dagger
\end{align}
The control operations $V_j$ that need to be implemented for the
decoupling scheme are of the form
$P_{1j} \otimes P_{2j} \otimes \ldots \otimes P_{nj}\in\cG^{\otimes
  n}$, where $\cG$ is a finite subgroup of $\cU(d)$.  Let us assume
that for each time step $j$ the local operations can be implemented in
parallel. When we convert the decoupling scheme into a controlization
scheme, however, we have to implement the control operations
\begin{align}
    \Lambda(V_j) 
    &= 
    |0\>\<0| \otimes P_{1j} \otimes P_{2j} \otimes \ldots \otimes P_{nj} +
    |1\>\<1| \otimes I \otimes I \otimes \ldots \otimes I.
\end{align}
Now, it becomes clear that the weight of $V_j$ (that is, the number of
non-identity components) impacts resources needed to realize
$\Lambda(V_j)$. Thus, one possible question is how to construct
decoupling schemes such that the control operations $V_j$ have small
weights. Here it may be helpful to consider orthogonal arrays
constructing from error correcting codes as in
\cite{bookatz2016improved}.

Another question related to this is how to realize the control gates
$\Lambda(P_{kj})$ acting on the control qubit and the $k$th qudit
efficiently since long-range gates may not always be directly
available. Methods such as those in \cite{baeumer2023efficient} may
help realize these gates.

\subsection*{Acknowledgements}

We would like to thank Patrick Rall for helpful discussions. 

\bibliographystyle{unsrt}
\bibliography{refs}

\end{document}